\documentclass[a4paper]{article}
\usepackage{a4wide}
\usepackage{authblk}
\usepackage{setspace}
\usepackage[pagewise]{lineno}
\usepackage{endfloat}

\usepackage{crop}%
\usepackage{amsmath,amssymb,amsfonts,upref,endnotes}
\usepackage{amsthm}
\usepackage{bm}
\usepackage{cite}
\usepackage{stmaryrd}
\usepackage{marginnote}
\usepackage{soul}
\RequirePackage{graphicx}
\usepackage[figuresright]{rotating}
\usepackage{color}
\makeatletter
\newcommand{\setlabel}[2]{\def\@currentlabel{#2}\label{#1}}
\makeatother
\newtheorem{theorem}{\bf Theorem}[section]

\newcommand{\pd}[2]{\frac{\partial #1}{\partial #2}}

\newcommand{\bv}[1]{\bm{#1}}

\usepackage{url}

\begin{document}


\newpage

\setcounter{page}{1}
\title{
Revisiting Volterra defects: Geometrical relation between edge dislocations and wedge disclinations
}

\author[1]{Shunsuke Kobayashi}
\author[1]{Katsumi Takemasa}
\author[1]{Ryuichi Tarumi}

\affil[1]{Graduate School of Engineering Science, Osaka University, 1-3 Machikaneyama-cho, Toyonaka, Osaka, 560-8531, Japan}
\date{}

\maketitle

\noindent

\begin{abstract}
    This study presents a comprehensive mathematical model for Volterra defects and explores their relations using differential geometry on Riemann--Cartan manifolds.
    Following the standard Volterra process, we derived the Cartan moving frame, a geometric representation of plastic fields, and the associated Riemannian metric using exterior algebra.
    Although the analysis naturally defines the geometry of three types of dislocations and the wedge disclination, it fails to classify twist disclinations owing to the persistent torsion component, suggesting the need for modifications to the Volterra process.
    By leveraging the interchangeability of the Weitzenb\"ock and Levi-Civita connections and applying an analytical solution for plasticity derived from the Biot--Savart law, we provide a rigorous mathematical proof of the long-standing phenomenological relationship between edge dislocations and wedge disclinations.
    Additionally, we showcase the effectiveness of novel mathematical tools, including Riemannian holonomy for analysing the Frank vector and complex potentials that encapsulate the topological properties of wedge disclinations as jump discontinuities.
    Furthermore, we derive analytical expressions for the linearized stress fields of wedge disclinations and confirm their consistency with existing results.
    These findings demonstrate that the present geometrical framework extends and generalizes the classical theory of Volterra defects.
\end{abstract}

\section{Introduction}
The regular atomic arrangement in a crystal structure is defined by 230 space groups, consisting of 32 point groups and seven Bravais lattices~\cite{aroyo_international_2016}.
These space groups determine the atomic arrangement of a perfect crystal.
The Bravais lattice defines the smallest unit, known as a primitive cell, that can fill the entire space through translational operations, and can be considered the crystal's skeletal framework.
A perfect crystal exhibits a state defined by a combination of symmetry operations~\cite{aroyo_international_2016}.
However, real crystals are not ideal and contain various structural irregularities known as lattice defects~\cite{mura_micromechanics_1987}.
Among them, one-dimensional line defects are particularly important.
These defects, known as Volterra defects, can be categorized into two types: dislocations and disclinations~\cite{volterra_sur_1907,anderson_theory_2017}.
Dislocations are associated with breaking the translational symmetry, whereas disclinations are associated with breaking rotational symmetry.
Both defect types disrupt the crystal lattice symmetry at the Bravais lattice level, significantly affecting the mechanical properties of the crystal structure.
Therefore, understanding the fundamental properties of lattice defects and achieving superior material designs through their control have been long-standing research topics in materials science and condensed-matter physics.

Theoretical analyses of Volterra defects have been conducted using differential geometry, particularly within Riemann--Cartan manifolds~\cite{katanaev_theory_1992,yavari_riemann--cartan_2012,yavari_riemann--cartan_2013,yavari_geometry_2014,golgoon_line_2018,sozio_geometric_2023,kobayashi_geometrical_2024,kobayashi_biot--savart_2024}.
A key advantage of this approach is its ability to decompose the kinematics into plastic and elastic deformations.
This multiplicative decomposition facilitates the geometric analysis of the plasticity.
For example, we recently identified the origin of the stress fields as geometrical frustration~\cite{kobayashi_geometrical_2024}.
Furthermore, the mathematical equivalence between Cartan's structure equations for plasticity, Amp\`ere's and Gauss' laws in electromagnetism, and the Cauchy--Riemann equations in complex function analysis has been elucidated~\cite{kobayashi_biot--savart_2024}.
This insight enables analytical integration and construction of complex potentials for dislocation plasticity.
In contrast to the recent advances in dislocation research, the study of disclinations, another type of Volterra defect, remains underdeveloped.
In contrast to the significant advancements in dislocation research, the study of disclinations remains relatively underdeveloped. This is partly because direct experimental observations have largely been limited to small crystals~\cite{galligan_fivefold_1972}, as disclinations typically require long-range stress fields~\cite{romanov_disclinations_1983,romanov_application_2009}.
Recent studies, however, have highlighted the emergence and significance of disclinations in deformation microstructures~\cite{li_disclination_1972,inamura_geometry_2019,egusa_recovery_2021,tokuzumi_role_2023,matsumura_numerical_2023,pranoto_evaluation_2024,zhang_origin_2024}, emphasizing their role as a strengthening mechanism in bulk materials. Clearly, a comprehensive understanding of both dislocations and disclinations, underpinned by modern differential geometry, is essential.
Early geometrical theories of disclinations were developed by Kondo~\cite{kondo_non-riemannian_1955}, Anthony~\cite{anthony_theorie_1970}, and Amari~\cite{amari_dualistic_1968,amari_dualistic_1981}, and these theories were shown to align with conventional defect theories through linearized analyses~\cite{dewit_theory_1973,de_wit_view_1981}. Interestingly, the stress fields of edge dislocations closely resemble those of wedge disclination dipoles~\cite{eshelby_simple_1966,li_disclination_1972}.
Similar connections have been reported from kinematic perspectives~\cite{dewit_partial_1972,kroner_dislocations_1975,hirth_disclinations_2020,marcinkowski_differential_1977,marcinkowski_relationship_1977,marcinkowski_differential_1978,kupferman_metric_2015,pretko_fracton-elasticity_2018}.
Despite these observations, however, to the best of the authors' knowledge, no rigorous mathematical proof has been established to substantiate this long-standing phenomenological hypothesis.
An intriguing perspective arises when examining the geometry of dislocations and disclinations through the mathematical framework of a Riemann--Cartan manifold.
This manifold consists of two classes, known as the Weitzenb\"ock and Riemannian manifolds.
Differential geometry distinguishes them according to their connections; the former includes torsion, whereas the latter includes curvature, which are regarded as mathematical representations of dislocations and disclinations, respectively~\cite{yavari_riemann--cartan_2012,yavari_riemann--cartan_2013,yavari_geometry_2014,anthony_theorie_1970,dzyaloshinskii_poisson_1980}.
However, as discussed later, the connection choice is not unique.
It is possible to interchange connections without altering the geometric states encoded in the Riemannian metric.
This mathematical arbitrariness of the connection provides a framework for unifying and classifying Volterra defects purely from a geometric perspective, which we believe will significantly advance the field of materials science.

In this study, a comprehensive mathematical model is developed for Volterra defects and their relationships are examined using differential geometry within the framework of Riemann--Cartan manifold.
The remaining of this paper is organized as follows.
Section 2 provides a brief overview of the mathematical foundations of the Riemann--Cartan manifold with a focus on the connections that play a pivotal role in the Volterra defect classification.
Section 3 introduces geometrical definitions of dislocations and disclinations based on the Volterra process.
Although the model naturally defines the geometry of the three types of dislocations and wedge disclinations, it fails to classify twist disclinations because of the persistent torsion component, indicating the need for modifications to the Volterra process.
Section 4 presents the core findings of this study.
This rigorously proves the long-standing phenomenological relationship between edge dislocations and wedge disclinations using Riemann--Cartan geometry.
Additionally, we demonstrate that an edge dislocation can be interpreted as the dipole moment of wedge disclinations.
Section 5 focuses on the mathematical analysis of disclinations.
First, the effectiveness of new mathematical tools, such as Riemannian holonomy for Frank vector evaluation and complex potentials for elucidating topological properties, including jump discontinuities, are highlighted.
Analytical expressions for linearized stress fields are derived and shown to quantitatively agree with existing results, confirming that the linearized geometric framework is fully consistent with previous studies.
Finally, Section 6 concludes the study.

\section{Mathematical foundations for Riemann--Cartan manifold}

\subsection{Cartan moving frame and Riemannian metric}

A Riemann--Cartan manifold is defined by the triplet $(\mathcal{M}, g, \nabla)$, which is a smooth manifold $\mathcal{M}$ equipped with a Riemannian metric $g$ and an affine connection $\nabla$~\cite{yavari_riemann--cartan_2012,yavari_geometry_2014}.
The metrics and connections are generalizations of the inner product and parallel transportation of vectors from Euclidean geometry on a manifold \cite{tu_differential_2017}.
For an $g$-orthonormal frame $\bv{e}_i$, the affine connection satisfies $\nabla_{\bv{e}_k}\bv{e}_j = \bv{\omega}^i_j(\bv{e}_k)\bv{e}_i$, where $\bv{\omega}^i_j=-\bv{\omega}^j_i$ is the connection 1-form.
We assumed that the affine connection $\nabla$ is compatible with the Riemannian metric $g$; that is $\nabla_{\bv{e}_i} g(\bv{e}_j, \bv{e}_k) = 0$.

The Volterra defect kinematics can be developed using the mathematical structure of the Riemann--Cartan manifold.
Three configurations are considered in this framework: reference, intermediate, and current~\cite{kobayashi_geometrical_2024,kobayashi_biot--savart_2024}.
Although these configurations share the same manifold $\mathcal{M}$, they differ in their metrics and connections.
The reference configuration represents a perfect crystal in Euclidean space, whereas the intermediate configuration requires Riemann--Cartan geometry to describe plastic deformation due to defects.
The current configuration is obtained by elastically embedding the intermediate configuration into Euclidean space.
The manifold $\mathcal{M}$ is assumed to be diffeomorphic to a subdomain of the three-dimensional Euclidean space $\mathbb{R}^3$.
Hereafter, unless stated otherwise, we employ rectangular coordinate system $\bv{x} = (x, y, z)$ and the dual basis $d\bv{x} = (dx, dy, dz)$ for the reference configuration.
Then, according to the Helmholtz decomposition, Cartan's moving frame $\bv{\vartheta} = (\bv{\vartheta}^1, \bv{\vartheta}^2, \bv{\vartheta}^3)$ on the intermediate configuration is expressed by a sum of the exact $d\bv{x}^i$ and dual exact forms $\bv{\Theta}^i$~\cite{kobayashi_geometrical_2024,kobayashi_biot--savart_2024}.
Explicitly, the moving frame and corresponding Riemannian metric are given by
\begin{align}
    \label{Eq:CoframeHelmholtzDecomposition}
    \bv{\vartheta}^i = d\bv{x}^i + \bv{\Theta}^i,
    \quad
    g = \delta_{ij} \bv{\vartheta}^i \otimes \bv{\vartheta}^j,
\end{align}
where $\delta_{ij}$ is the Kronecker delta.
Hence, Riemannian metric $g$ can be determined from the moving frame $\bv{\vartheta}^i$.

\subsection{Cartan structure equations and affine connections}

Another mathematical component of the Riemann--Cartan manifold is the affine connection $\nabla$, which incorporates geometric features, such as torsion $\bv{\tau}^i$ and curvature $\bv{\Omega}^i_j$.
Previous studies have suggested that torsion and curvature correspond to dislocations and disclinations, respectively~\cite{kondo_non-riemannian_1955,bilby_continuous_1955,kroner_allgemeine_1959,anthony_theorie_1970}.
Therefore, the connection plays a crucial role in the geometric analysis of Volterra defects.
In the standard differential geometry framework, Cartan's moving frame $\bv{\vartheta}^i$ and corresponding connection 1-form $\bv{\omega}^i_j$ are related by Cartan's first and second structure equations, which are expressed as follows:
\begin{align}
    \label{eq:CartanEquations}
    \bv{\tau}^i = d\bv{\vartheta}^i + \bv{\omega}^i_j \wedge \bv{\vartheta}^j,
    \quad
    \bv{\Omega}^i_j = d \bv{\omega}^i_j + \bv{\omega}^i_k \wedge \bv{\omega}^k_j,
\end{align}
where $d$ denotes the exterior derivative.
These equations relate the derivative of the moving frame $d\bv{\vartheta}^i$ to the torsion $\bv{\tau}^i$ and curvature $\bv{\Omega}^i_j$ through connection $\bv{\omega}^i_j$.
A key observation is that the connection choice is not unique; in fact, there are countless affine connections that satisfy Cartan's structure equations.
One possible choice is to assume that the derivative of moving frame $d\bv{\vartheta}^i$ is entirely due to torsion by simply setting the connection 1-form $\bv{\omega}^i_j$ to zero.
This construction is known as the Weitzenb\"ock connection $\nabla^W$.
In this case, Cartan's structure equations (\ref{eq:CartanEquations}) become
\begin{align}
    \label{eq:CartanEquationsTau}
    \bv{\tau}^i= d\bv{\vartheta}^i,
    \quad
    \bv{\Omega}^i_j=0.
\end{align}
The Weitzenb\"ock connection $\nabla^W$ includes non-zero torsion while maintaining a vanishing curvature~\cite{wenzelburger_kinematic_1998,yavari_riemann--cartan_2012}.
The Riemann--Cartan manifold with the Weitzenb\"ock connection $(\mathcal{M}, g, \nabla^W)$ is called Weitzenb\"ock manifold.
This mathematical construction serves as a geometric model for dislocations~\cite{yavari_riemann--cartan_2012,yavari_geometry_2014}.
Another important choice is the Levi-Civita connection, $\nabla^L$.
In this case, Cartan's structure equations are satisfied when the derivative $d\bv{\vartheta}^i$ increases solely from the curvature with vanishing torsion.
Consequently, Cartan's structure equations (\ref{eq:CartanEquations}) become
\begin{align}
    \label{Eq:CartanEquationsCurv}
    0 = d\bv{\vartheta}^i + \bv{\omega}^i_j \wedge \bv{\vartheta}^j,
    \quad
    \bv{\Omega}^i_j = d\bv{\omega}^i_j + \bv{\omega}^i_k \wedge \bv{\omega}^k_j.
\end{align}
The Riemann--Cartan manifold $(\mathcal{M}, g, \nabla^L)$ equipped with a Levi-Civita connection is called Riemannian manifold.
Riemannian manifolds are considered geometric models of disclinations~\cite{anthony_theorie_1970,yavari_riemann--cartan_2013,yavari_geometry_2014}.

\section{Geometric definition of Volterra defects}

\subsection{Continuous deformation for the Volterra process}

Figure \ref{fig:VolterraDefects} schematically illustrates the six types of Volterra defects~\cite{volterra_sur_1907}.
The cylinder is cut along the $z$-axis, where one of the cut surfaces remains fixed and the other is displaced relative to it, representing the plastic deformation caused by a Volterra defect in a perfect crystal.
The translational displacements of the cut surface along the $x$-, $y$-, and $z$-axes are shown in figures 1(a)--(c), respectively.
These displacements represent dislocations, with their magnitudes denoted by the Burgers vector $\bv{b}$.
If we consider the infinitesimal limit of the cylinder radius, the dislocation line corresponds to the $z$-axis.
Similarly, rotational displacements are applied around the $x$-, $y$-, and $z$-axes, resulting in figures 1(d)--(f), respectively.
These rotational displacements represent disclinations, and the angle $\phi$ is understood as the Frank vector.
These six types of defects are known as Volterra defects and the process of introducing defects into a perfect crystal is referred to as the Volterra process~\cite{volterra_sur_1907}.
It is important to note that the Volterra process is purely plastic, with no elastic deformations involved.

\begin{figure}[!h]
  \centering
  \includegraphics[width=0.95\textwidth]{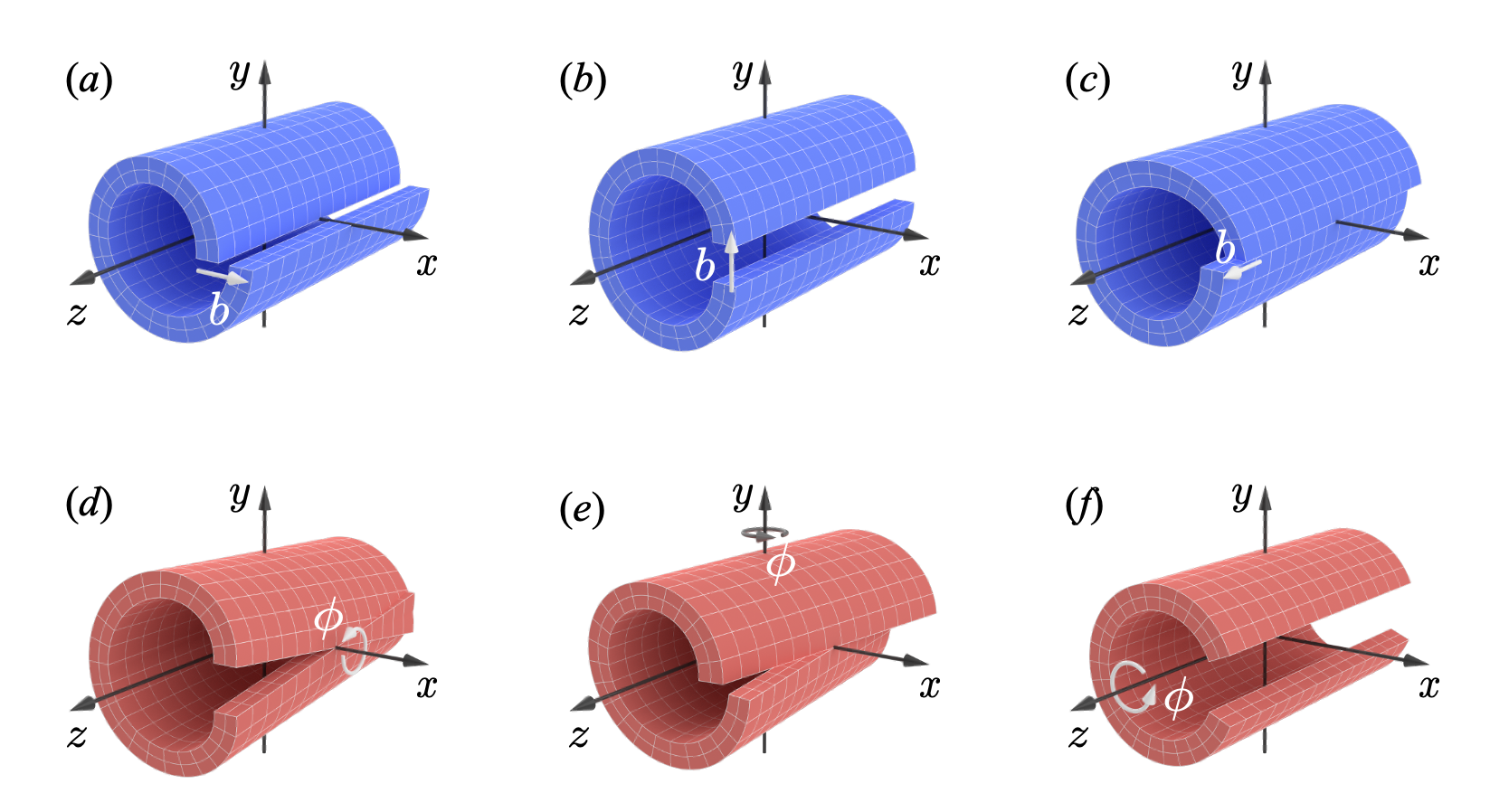}
  \caption{
    Schematic  of the formation of dislocations and disclinations using the Volterra process.
    Here, the lattice defect is aligned with the $z$-axis.
    The translational displacements perpendicular to the $z$-axis in (a) and (b) correspond to edge dislocations, while the translational displacement parallel to the $z$-axis in (c) corresponds to a screw dislocation.
    Rotations about an axis perpendicular to the $z$-axis (as seen in (d) and (e)) are called twist disclinations, whereas a rotation around an axis parallel to the $z$-axis (as seen in (f)) is referred to as a wedge disclination.  
  }
  \label{fig:VolterraDefects}
\end{figure}

We used differential geometry to define lattice defects introduced by the Volterra process.
Let us express the Volterra process by using a continuous deformation $\bv{\psi}=\bv{\psi}(\bv{x})$ excluding the defect line placed along the $z$-axis.
The Volterra process involves translational and rotational plastic deformations, as shown in figures \ref{fig:VolterraDefects}.
Therefore, its general form is given by~\cite{puntigam_volterra_1997}
\begin{align}
\label{eq:VolterraTranslationRotation}
    \bv{\psi}(\bv{x}) = \mathcal{R}(\bm{x})\bm{x} + \mathcal{T}(\bm{x}),
\end{align}
where $\mathcal{T} \in T(3)$ and $\mathcal{R} \in SO(3)$ represent the three-dimensional translation and rotation, respectively.
The derivative of the Volterra deformation, expressed as $\bv{F}_p=\nabla \bv{\Psi}$, is referred to as the plastic deformation gradient.
It should be noted that, in general, $\bv{F}_p$ cannot be represented by the gradient of a function; however, this is possible if the defects are localised on a linear or planar subdomain in $\mathbb{R}^3$ and exclude that domain from the analysis.
The Volterra deformation defines a linear map from the standard Euclidean frame $d\bm{x}$ to the Cartan frame $\bv{\vartheta}$ in such a way that $\bv{\vartheta}(\bv{x})=\bv{F}_p d\bv{x}=(\bv{I}+\nabla \bv{u}_p)d\bv{x}=d\bv{x} + \bv{\Theta}$, where $\bv{I}$ represents the $3\times 3$ identity matrix.
Comparing the result with equation (\ref{Eq:CoframeHelmholtzDecomposition}), it is obvious that the dual-exact form $\bm{\Theta}$ of the moving frame represents the change in the Euclidean frame due to the gradient of plastic displacement $\nabla \bm{u}_p$.
To simplify the analysis, we assumed that the translation and rotation of the Volterra process occurred independently rather than simultaneously.
In the case of a dislocation, we use a non-zero translation $\mathcal{T}$ and identity rotation $\mathcal{R}=\bm{I}$.
Consequently, the Volterra deformation and corresponding Cartan frame are given by
\begin{align}\label{eq:VolterraTranslation}
    \bv{\psi}(\bm{x}) = \bm{x} + \mathcal{T}(\bm{x}), \quad \bm{\vartheta}(\bm{x}) = d\bm{x} + d\mathcal{T}(\bm{x}),
\end{align}
where $d\mathcal{T} = d(\mathcal{T}(\bm{x}))$ is the exterior derivative of the translational deformation.
By contrast, disclinations involve only non-zero rotation $\mathcal{R}$ with a vanishing translation.
Consequently, the Volterra deformation and corresponding Cartan frame are given by
\begin{align}
    \label{Eq:DisclinationPsi}
    \bv{\psi}(\bm{x})=\mathcal{R}(\bm{x})\bm{x},
    \quad
    \bm{\vartheta}'=d\bm{x}+\mathcal{R}^{-1}(d\mathcal{R})\bm{x},
\end{align}
where $\bm{\vartheta}' = \mathcal{R}^{-1} \bm{\vartheta}$.
Note that $\mathcal{R}^{-1}$ represents the inverse matrix of the rotation and $d\mathcal{R} = d(\mathcal{R}(\bm{x}))$.
Note that the Riemannian metric $g$ remains unchanged owing to the local rotation by $\mathcal{R}^{-1}$.

\subsection{Edge dislocations}

The translational displacement of edge dislocations occurs in the direction perpendicular to the dislocation line.
As shown in figures \ref{fig:VolterraDefects}(a) and (b), this displacement can occur in the $x$- or $y$-direction.
However, because these two cases are essentially equivalent, we focused on analysing the edge dislocation with the Burgers vector in the $x$-direction, as illustrated in figure \ref{fig:VolterraDefects}(a).
The translational displacement of the dislocation can be described by the following mapping:
\begin{align}
    \label{Eq:EdgeDislocationVolterraMap}
    \bv{\psi}(\bm{x})
    = \bm{x} + \frac{b}{2\pi} \arctan \left( \frac{y}{x} \right)(1, 0, 0).
\end{align}
The second term on the right-hand side represents the translational displacement $\mathcal{T}(\bm{x})$ in the $x$-direction, and $b/2\pi$ is a normalization coefficient with the magnitude of Burgers vector $b$.
It is convenient to introduce cylindrical coordinates $(r, \theta, z)$, where $r = \sqrt{x^2 + y^2}$ and $\theta = \arctan (y/x)$ on a plane normal to the dislocation line.
By taking the exterior derivative of the above equation, we obtain the Cartan frame and corresponding Riemannian metric for the edge dislocation in the cylindrical coordinate, such that
\begin{align}
    \label{Eq:EdgeDislocationCoframe}
    \bv{\vartheta}=d\bv{x}+\frac{b}{2\pi}\left(d\theta,\, 0,\, 0 \right),
    \quad
    g=\begin{pmatrix}
        1 & \frac{b}{2\pi}\cos{\theta} & 0\\
        \frac{b}{2\pi}\cos{\theta} & r^2-\frac{b}{\pi}r\sin{\theta}+\frac{b^2}{4\pi^2} & 0\\
        0 & 0 & 1
    \end{pmatrix},
\end{align}
where $d\theta$ represents the exterior derivative of polar angle $\theta$, which is expressed as
\begin{align}
    \label{Eq:dtheta}
    d\theta = d\left(\arctan \frac{y}{x} \right) = -\frac{y}{x^2 + y^2} dx + \frac{x}{x^2 + y^2} dy.
\end{align}
The Riemannian metric $g$ provided in equation (\ref{Eq:EdgeDislocationCoframe}) aligns with our previous result which was obtained through integration of Cartan's first structure equation~\cite{kobayashi_biot--savart_2024}.
By determining the connection $\nabla$ associated with the moving frame $\bv{\vartheta}$, the Volterra defect can be characterized as a Riemann--Cartan manifold.
According to equation (\ref{Eq:EdgeDislocationCoframe}), the reference frame $d\bv{x}$ undergoes a continuous displacement in the $x$-direction.
Cartan structure equations (\ref{eq:CartanEquations}) represent the change in the frame as $d\bv{\vartheta}$ and relates it to the torsion $\bv{\tau}^i$ and connection $\bv{\omega}^i_j$.
However, as long as the structure equations are satisfied, the proportion of the change in $d\bv{\vartheta}$ distributed between $\bv{\tau}^i$ and $\bv{\omega}^i_j$ remains arbitrary.
One possible choice is Weitzenb\"ock connection (\ref{eq:CartanEquationsTau}), where the torsion fully accounts for the derivative of the moving frame, $d\bv{\vartheta}^i = \bv{\tau}^i$, by setting $\bv{\omega}^i_j = 0$.
This results in a Weitzenb\"ock manifold.

As shown in equation (\ref{Eq:EdgeDislocationCoframe})$_1$, the plastic deformation of an edge dislocation affects only the component $\bv{\vartheta}^1$, while the other components, $\bv{\vartheta}^2$ and $\bv{\vartheta}^3$, remain identical to those of the reference state.
Consequently, their exterior derivatives vanish, resulting in $\bv{\tau}^2 = \bv{\tau}^3 = 0$.
In contrast, applying Gauss's divergence theorem to the $\bv{\vartheta}^1$ component yields $d(d\theta) = 2\pi \delta(x)\delta(y) dx \wedge dy$, indicating non-zero torsion along the $z$-axis.
Therefore, the corresponding torsion 2-form is:
$\bv{\tau}^1=b\delta (x) \delta (y)dx\wedge dy$.
According to previous studies, the dislocation density tensor $\bv{\alpha}$ can be expressed as $\bv{\alpha} = *\bv{\tau}$, where the Hodge star operation is applied to the torsion 2-form~\cite{yavari_riemann--cartan_2012}.
Combining these results, we obtain
\begin{align}
    \label{Eq:EdgeDislocationDensity}
    \bm{\tau} = b \big(\delta(x, y) dx \wedge dy, 0, 0\big),
    \quad
    \bm{\alpha} = *\bm{\tau}=b \big(\delta(x, y) dz, 0, 0\big).
\end{align}
This shows that an edge dislocation with a Burgers vector of magnitude $b$ exists along the $z$-axis, which is consistent with the classical definition of dislocation theory.
Additionally, consider an arbitrary closed circuit $C$ encircling the origin with surface $A$ bounded by $C$, that is, $C = \partial A$.
Stokes' theorem for differential forms leads to the following relationship:
\begin{align}
    \int_C \bv{\vartheta}^1=
    \int_C \left(dx + \frac{b}{2\pi} d\theta\right)
    = \int_A b \delta(x, y) dx \wedge dy = b.
\end{align}
This confirms that $C$ is a Burgers circuit, which is consistent with the classical dislocation theory.    
For the edge dislocation shown in figure \ref{fig:VolterraDefects}(b), the only difference is the translation $\mathcal{T}$ in equation (\ref{Eq:EdgeDislocationVolterraMap}), which is now parallel to the $y$-axis, while the rest of the analysis remains essentially the same.

\subsection{Screw dislocation}

The Volterra process for the screw dislocation shown in figure \ref{fig:VolterraDefects}(c) is described via plastic deformation to the $z$-axis direction in such a way that
\begin{align}
    \label{eq:ScrewDislocationVolterraMap}
    \bv{\psi}(\bm{x}) = \bm{x} + \frac{b}{2\pi} \arctan \left(\frac{y}{x}\right)(0, 0, 1).
\end{align}
This transformation is smooth, except along the $z$-axis, which corresponds to the dislocation line.
For the Volterra deformation of the screw dislocation given in equation (\ref{eq:ScrewDislocationVolterraMap}), a direct calculation yields the analytical expression of the moving frame $\bv{\vartheta}$ and associated Riemannian metric $g$ in the cylindrical coordinate such that
\begin{align}
    \label{Eq:ScrewDislocationCoframe}
    \bv{\vartheta}=d\bv{x}
    +\frac{b}{2\pi}\left( 0,\, 0,\, d\theta \right),
    \quad
    g=\begin{pmatrix}
        1 & 0 & 0\\
        0 & r^2+\frac{b^2}{4\pi^2}& \frac{b}{2\pi}\\
        0 & \frac{b}{2\pi} & 1
    \end{pmatrix}.
\end{align}
Inserting equations (\ref{Eq:ScrewDislocationCoframe}) and (\ref{Eq:dtheta}) into Cartan's structure equations (\ref{eq:CartanEquationsTau}) with the Weitzenb\"ock connection, the torsion 2-form and corresponding dislocation density become
\begin{align}
    \label{Eq:ScrewDislocationDensity}
    \bm{\tau} = b \big(0,\, 0,\, \delta(x, y) dx \wedge dy\big), \quad
    \bm{\alpha}= b \big(0,\, 0,\, \delta(x, y) dz\big).
\end{align}
This result is consistent with the classical screw dislocation definition.
Any closed circuit $C$ encircling the origin leads the Burgers vector $b$ of the screw dislocation.

\subsection{Wedge disclination}

Similar to the systematical analysis to dislocations, we now examine the Volterra process of disclinations.
As shown in figure \ref{fig:VolterraDefects}(f), the plastic deformation of a wedge disclination is characterized by a rotation $\phi$ about the $z$-axis. The corresponding Volterra deformation can be expressed by:
\begin{align}
    \label{Eq:VolterraWedge}
    \psi(\bm{x})=\left(
    \begin{array}{ccc}
        \cos \frac{\phi}{2\pi}\theta & -\sin \frac{\phi}{2\pi}\theta & 0\\
        \sin \frac{\phi}{2\pi}\theta & \cos \frac{\phi}{2\pi}\theta & 0\\
        0 & 0 & 1
    \end{array}\right)
    \begin{pmatrix}
        x\\
        y\\
        z
    \end{pmatrix}.
\end{align}
The $3 \times 3$ matrix $\mathcal{R}(\bv{x})$ acting on $\bv{x}=(x, y, z)$ represents the rotation around the $z$-axis due to the Volterra process.
The rotation angle is proportional to $\theta$ and the maximum angle corresponds to the Frank vector $\phi$.
By considering the exterior derivative of the Volterra deformation (\ref{Eq:VolterraWedge}), the moving frame $\bm{\vartheta}'$ and associated Riemannian metric $g$ in the cylindrical coordinate can be expressed as follows:
\begin{align}
    \label{Eq:WedgeDisclinationCoframe}
    \bv{\vartheta}'=d\bv{x}+\frac{\phi}{2\pi}
    \left(- y d\theta,\, x d\theta,\, 0 \right),\quad
    g=\begin{pmatrix}
        1 & 0 &0\\
        0 & \left(1+\frac{\phi}{2\pi}\right)^2r^2 & 0
        \\
        0 & 0 & 1
    \end{pmatrix}.
\end{align}
The Riemannian metric $g$ in the above equation aligns with the previous study~\cite{yavari_riemann--cartan_2013}.
Next, we determine the connection for the frame.
As discussed in previous sections, we used the Weitzenb\"ock connection for the analysis of dislocations. 
However, for the disclination analysis, we adopt the Levi-Civita connection.
This implies that the torsion 2-form is assumed to be zero, that is, $\bv{\tau}^i = 0$, which means that we construct a curved manifold.
To simplify the expressions, we introduce representations for connection 1-form
$\bv{\bv{\omega}}=(\bv{\omega}^2_3, \bv{\omega}^3_1, \bv{\omega}^1_2)$ and corresponding curvature $\bv{\Omega}=(\Omega^2_3, \Omega^3_1, \Omega^1_2)$ obtained from Cartan's second structure equation.
In the case of a wedge disclination we have
\begin{align}
    \label{eq:WedgeCurvatureForms}
    \bv{\bv{\omega}}=\frac{\phi}{2\pi}\left( 0,\, 0,\, -d\theta \right),
    \quad
    \bv{\Omega}=\phi \left( 0,\, 0,\, -\delta(x,y)dx\wedge dy \right).
\end{align}
As $\bv{\omega}^i_k \wedge \bv{\omega}^k_j$ is zero, the curvature form $\bv{\Omega}^i_j$ is given by the exterior derivative of connection form $\bv{\omega}^i_j$.
The non-zero curvature form $\bv{\Omega}^1_2$ signifies the Frank vector, with a magnitude $\phi$, is oriented along the $z$-axis.
This result implies that the disclination line aligns with the $z$-axis.
This is the geometrical definition of a wedge disclination.

\subsection{Twist disclinations}
Similarly, the plastic deformation of a twist disclination is formed by applying a rotation to an axis perpendicular to the disclination line (see figures \ref{fig:VolterraDefects}(d) and (e)).
There are two possible choices for this rotation axis: the $x$- and $y$- axes.
However, because these two cases are essentially equivalent, we proceed with the analysis of the rotation about the $y$-axis.
In this case, the Volterra process can be described using the following continuous deformation:
\begin{align}\label{Eq:VolterraTwist}
    \psi(\bv{x})=
    \begin{pmatrix}
        \cos \frac{\phi}{2\pi}\theta & 0 & \sin \frac{\phi}{2\pi}\theta \\
        0 & 1 & 0\\
        -\sin \frac{\phi}{2\pi}\theta & 0 & \cos \frac{\phi}{2\pi}\theta\\
    \end{pmatrix}
    \begin{pmatrix}
        x\\
        y\\
        z
    \end{pmatrix}.
\end{align}
The Cartan frame $\bm{\vartheta}'$ and associated Riemannian metric $g$ in the cylindrical coordinate resulting from this mapping is calculated as follows:
\begin{align}
    \label{Eq:TwistDisclinationCoframe}
    \bv{\vartheta}'
    &=d\bv{x}+\frac{\phi}{2\pi}
    \left(z d\theta,\, 0,\, -xd\theta \right),\\
    g&=\begin{pmatrix}
        1 & \frac{\phi}{2\pi}z\cos{\theta} & 0\\
        \frac{\phi}{2\pi}z \cos{\theta}& r^2-\frac{\phi}{\pi}zr\sin{\theta}+\frac{\phi^2}{4\pi^2}(r^2\cos^2{\theta}+z^2) & -\frac{\phi}{2\pi} r\cos\theta\\
        0 & -\frac{\phi}{2\pi} r\cos\theta & 1
    \end{pmatrix}.
\end{align}
We can derive an affine connection $\bv{\omega}^i_j$ for the moving frame $\bv{\vartheta}'$ to construct a Riemannian--Cartan manifold with a twist disclination.
Similar to the wedge disclination case, we obtain the connection as follows:
\begin{align}
    \label{Eq:TwistCurvatureForms}
    \bv{\bv{\omega}}=\frac{\phi}{2\pi}\left( 0,\, -d\theta,\, 0 \right),
    \quad
    \bv{\Omega}=\phi \left( 0,\, -\delta(x,y)dx\wedge dy,\, 0 \right).
\end{align}
If we compare this curvature form $\bv{\Omega}$ with the analytical results for the wedge disclination expressed in equation (\ref{eq:WedgeCurvatureForms}), we can see that the curvature along the rotation axis is identical in both cases, and the magnitude corresponds to Frank vector $\phi$.
At first glance, it seems that we have successfully defined the twist disclination geometrically using the Volterra process given by equation (\ref{Eq:TwistDisclinationCoframe}).
However, there is a significant issue. In fact, the change in coframe $d\bv{\vartheta}'$ described by equation (\ref{Eq:TwistDisclinationCoframe}) cannot be fully captured by connection form $\bv{\omega}$ alone, leaving the following torsion unresolved:
\begin{align}\label{twisttorsion}
    \bv{\tau}=\phi z(\delta(x,y)dx\wedge dy,\, 0,\,0).
\end{align}
This result indicates that the Volterra process shown in figure \ref{fig:VolterraDefects}(e) cannot produce a pure twist disclination.
Instead, it generates an additional edge dislocation whose dislocation line aligns with the $z$-axis and has non-constant Burgers vector $\phi z$ in the $x$-axis direction.
Specifically, the Burgers vector increases proportionally to the magnitude of the Frank vector and its position along the $z$-axis.
A similar issue arises with the twist disclination expressed in figure \ref{fig:VolterraDefects}(d).
By contrast, this problem does not occur for the wedge disclination.
As indicated by equation (\ref{Eq:WedgeDisclinationCoframe}), the rotation axis of the Frank vector for the wedge disclination aligns with the dislocation line, avoiding the issue found in twist disclinations.
This mathematical result for the twist disclinations is highly inconsistent with a natural expectation that it is a purely curvature-type defect.
This discrepancy suggests that twist disclinations cannot be properly constructed using the Volterra processes illustrated in figures \ref{fig:VolterraDefects}(d) and (e).
We can obtain the essentially same conclusions by considering the Levi-Civita connection.
 In this case, the torsion becomes zero, while the curvature $\bv{\Omega}$ does not represent a pure twist disclination and other components including wedge disclination remain.
Therefore, we have to reconsider the Volterra process to introduce twist disclinations.

\section{Equivalence of edge dislocations and wedge disclinations}

\subsection{Wedge disclination dipole}
Mathematically, any affine connection $\nabla$ defines a Riemann--Cartan manifold $(\mathcal{M}, g, \nabla)$~\cite{yavari_geometry_2014}.
In this framework, changing an affine connection does not alter the geometry of Riemann--Cartan manifolds as long as they share the same Riemannian metric $g$.
More specifically, replacing the affine connection has no impact on either plastic or elastic deformations.
For instance, the example above demonstrates that two distinct manifolds, namely, the Weitzenb\"ock $(\mathcal{M}, g, \nabla^W)$ and Riemannian manifolds $(\mathcal{M}, g, \nabla^L)$, can both arise from the same moving frame $\bv{\vartheta}$ by selecting different affine connections.
This indicates that one connection can be freely replaced by another while maintaining the plastic deformation encoded in $\bv{\vartheta}$.
However, the inherent mathematical arbitrariness in the connection choice is of significance in the Volterra defect theory.
This is because the two manifolds, Weitzenb\"ock and Riemannian, can be understood as mathematical representations of different types of Volterra defects: dislocations and disclinations.
This seemingly contradictory conclusion indicates that the six types of Volterra defects are not geometrically independent.
Although similar suggestions have been made in several reports~\cite{li_disclination_1972,marcinkowski_relationship_1977,marcinkowski_differential_1978,marcinkowski_dislocations_1990}, there is no rigorous mathematical proof, as the analytical form of the plastic deformation fields $\bv{\vartheta}$ has not been elucidated.
Recently, we revealed the mathematical equivalence between Cartan's first structure equation for dislocations and Maxwell's static equations in electromagnetics~\cite{kobayashi_biot--savart_2024}.
This unexpected discovery enabled us to obtain an analytical expression for $\bv{\vartheta}$ formed around dislocations using the Biot--Savart law.
By combining this mathematical method with the geometric definition of Volterra defects presented in the previous section, we now provide mathematical proof of the long-standing phenomenological hypothesis that edge dislocations and wedge disclinations are geometrically equivalent.

As shown in figure \ref{fig:bipolar coordinate}(a), we considered a straight edge dislocation array with Burgers vector $\bv{b} = (0, b, 0)$.
This array is aligned along the $x$-axis over the finite interval $-L^- < x < L^+$ for $L^\pm>0$.
From a crystallographic perspective, this configuration is known as a symmetrical tilt boundary~\cite{hull_introduction_2011}.
It has long been hypothesized that wedge disclination dipoles are present at both ends of the dislocation array~\cite{li_disclination_1972}, although a rigorous mathematical proof has not been provided.
According to the geometrical definition of an edge dislocation given in equation (\ref{Eq:EdgeDislocationDensity}), the distribution of non-vanishing torsion caused by the edge dislocation array can be expressed as follows:
\begin{align}
    \label{Eq:EdgeDislocationArrayTorsion}
    \bm{\tau}^2
    =b \rho \delta(y) \big( H(x+L^-)-H(x-L^+) \big) dx\wedge dy,
\end{align}
where $\rho = N/(L^-+L^+)$ represents the number density of edge dislocations, and $N$ is the number of edge dislocations within the interval.
The symbols $\delta$ and $H$ denote the Dirac delta and Heaviside step functions, respectively.
Then, the following theorem holds:

\begin{figure}[!h]
  \centering
  \includegraphics[width=0.95\textwidth]{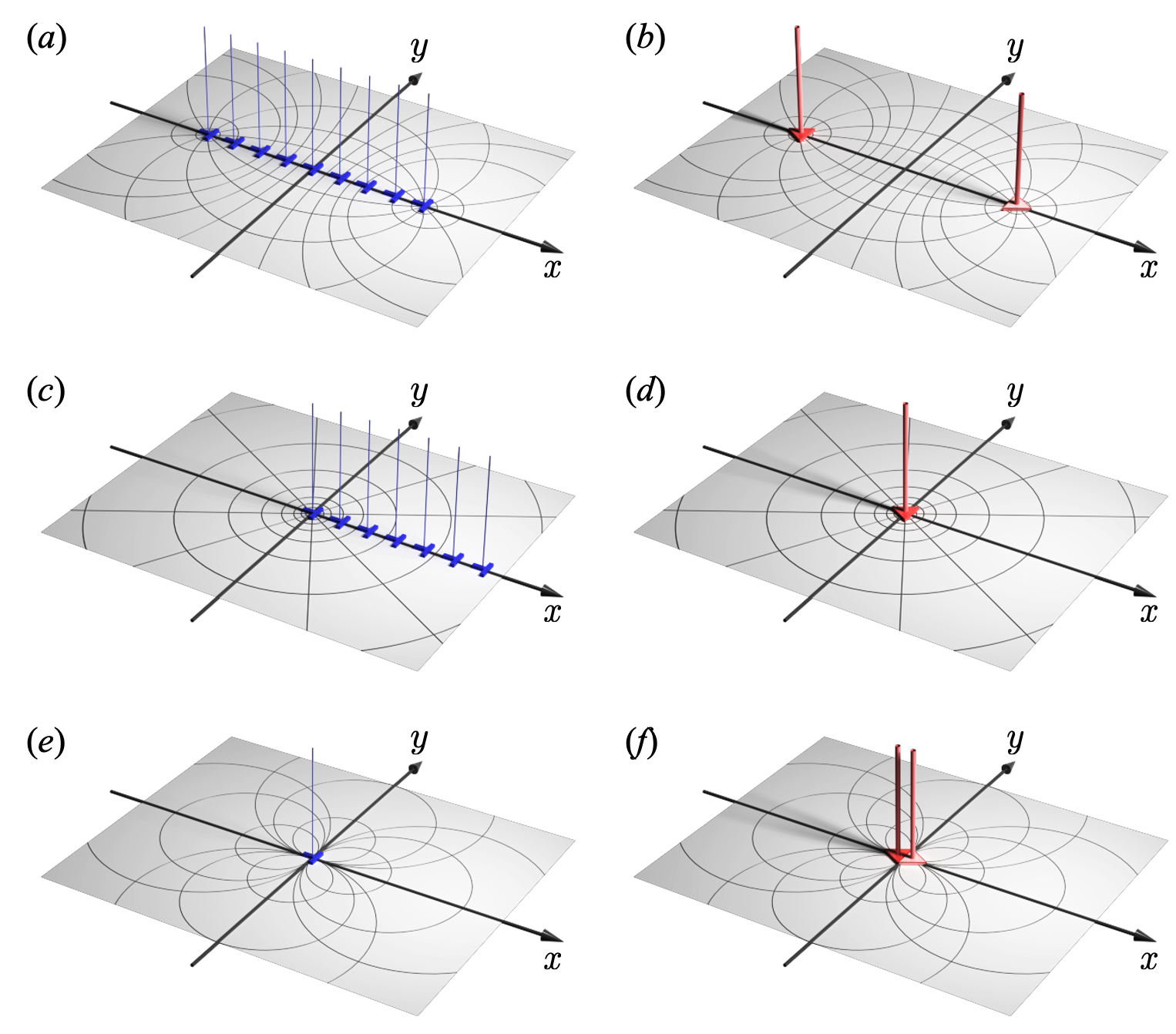}
  \caption{
  Equivalence of plastic deformation fields formed by edge dislocations and wedge disclinations.
  (a) Straight edge dislocation array and (b) equivalent wedge disclination dipole configuration.
  (c) Semi-infinite edge dislocation array and (d) equivalent wedge disclination monopole configuration.
  (e) Single-edge dislocation and (f) wedge disclination dipole with infinitesimal distance.
  The equi-contour curves for (a) and (b) are in bipolar coordinate, (c) and (d) are in log-polar coordinate, and (e) and (f) are point dipoles, which are given by the infinitesimal limit $L^\pm \to 0$ of the bipolar coordinate.
  }
  \label{fig:bipolar coordinate}
\end{figure}

\begin{theorem}[Wedge disclination dipole]
\label{thm:Dipole}
Suppose we have a straight array of edge dislocations whose torsion is given by equation (\ref{Eq:EdgeDislocationArrayTorsion}).
Then, a wedge disclination dipole exists at the terminal points of the dislocation array.
Moreover, the Frank vectors of the wedge disclinations are $\phi=\pm b\rho$.
\end{theorem}

\begin{proof}
First, we determine the plastic deformation fields formed by the edge dislocations using the Biot--Savart law.
This is based on the mathematical equivalence of Cartan's first structure equation for plasticity and Amp\`ere's and Gauss' law in electromagnetics~\cite{kobayashi_biot--savart_2024}.
More precisely, the plastic displacement gradient $\bv{\Theta}^i$ of dislocations can be calculated from the dislocation density $\bv{\alpha}^i$ in the following form:
\begin{align}
  \label{Eq:Biot-Savart}
  \bv{\Theta}^i(\bv{x}) = \frac{1}{4\pi} \int_{\mathbb{R}^3}
  \frac{\bv{\alpha}^i(\bv{\xi})\times (\bv{x}-\bv{\xi})}{\|\bv{x}-\bv{\xi}\|^3} dV.
\end{align}
By inserting the Hodge dual of the torsion $\bm{\tau}^2$ given in equation (\ref{Eq:EdgeDislocationArrayTorsion}) into (\ref{Eq:Biot-Savart}), we obtain the analytical expression for the plastic displacement gradient as follows:
\begin{align}
\label{Eq:EdgeDislocationArrayCoframe}
    \bv{\Theta}^2
    =\frac{b\rho}{2\pi}\left(
    \mathrm{arg}\biggl(\frac{y}{x+L^-}\biggr)
    -\mathrm{arg}\biggl(\frac{y}{x-L^+}\biggr),\ 
    \frac{1}{2}
    \ln\frac{(x+L^-)^2+y^2}{(x-L^+)^2+y^2},\ 0
    \right),
\end{align}
where $\bv{\Theta}^1=\bv{\Theta}^3=\bv{0}$ and $\mathrm{arg}(y/x)$ are defined by using the inverse tangent $\arctan(y/x)$ as follows:
\begin{align}
    \mathrm{arg}\biggl(\frac{y}{x}\biggr)=\begin{cases}
        0 & y=0\\
        \arctan(y/x) & x\leq 0\\
        \arctan(y/x)-\pi & x>0,\ y>0\\
        \arctan(y/x)+\pi & x>0,\ y<0\\
    \end{cases}.
\end{align}
By definition, $-\pi<\mathrm{arg}(y/x)<\pi$ represents the polar angle of a point $(x,y)$ measured from the negative $x$-axis except for the discontinuous singularity along $y=0\, (x> 0)$.
According to the Helmholtz decomposition (\ref{Eq:CoframeHelmholtzDecomposition}), the moving frame becomes $\bv{\vartheta}^2=(dx, dy+\bv{\Theta}^2, dz)$.
A direct calculation shows that the moving frame satisfies Cartan's structure equations (\ref{eq:CartanEquationsTau}) with the Weitzenb\"ock connection.

We expressed the same plastic deformation encoded in $\bv{\vartheta}^2$ using wedge disclination, rather than edge dislocations.
As discussed in section 2, this can be achieved by replacing the Weitzenb\"ock connection with the Levi-Civita connection.
To this end, we insert the moving frame $\bv{\vartheta}^2$ into Cartan's structure equations (\ref{Eq:CartanEquationsCurv}) with the Levi-Civita connection.
From the first structure equation (\ref{Eq:CartanEquationsCurv})$_1$, we obtain
\begin{align}
    \bv{\omega}^1_2\wedge (dy+\bv{\Theta}^2)-\bv{\omega}^3_1\wedge dz=0,\quad
    \bv{\omega}^1_2\wedge dx-\bv{\omega}^2_3\wedge dz=d\bv{\vartheta}^2,\quad
    \bv{\omega}^3_1\wedge dx -\bv{\omega}^2_3\wedge (dy+\bv{\Theta}^2)=0.
\end{align}
Consequently, we have the following non-vanishing connection form:
\begin{align}\label{eq:Curvature2FormFiniteArray}
     \bv{\omega}^1_2=-b \rho \delta(y) \big( H(x+L^-)-H(x-L^+) \big)dy,
     \quad
     \bv{\omega}^2_3=\bv{\omega}^3_1=0.
 \end{align}
By substituting the corresponding result into the second structure equation (\ref{Eq:CartanEquationsCurv})$_2$, we obtain the non-vanishing curvature responsible for the plastic deformation field $\bv{\vartheta}^2$:
\begin{align}
    \bv{\Omega}^1_2=b\rho \big( \delta(x-L^+,y)-\delta(x+L^-,y) \big) dx\wedge dy.
\end{align}
Note that we used the fundamental relation between the Heviside step function and Dirac delta function: $\frac{d}{dx}H(x)=\delta(x)$.
Following the geometric definition of the wedge disclination (\ref{eq:WedgeCurvatureForms}), we can conclude on the existence of a pair of wedge disclinations at $\bv{x}^+=(L^+,\, 0)$ with Frank vector $\phi= b\rho$
and at $\bv{x}^-=(-L^-,\, 0)$ with $\phi=-b\rho$.
\end{proof}

The geometric equivalence between the edge dislocation array and wedge disclination dipole can be readily understood by introducing a bipolar coordinate system.
To simplify the analysis, let us set $L^- = L^+ = L$.
It is well known that bipolar coordinates use two components $(\sigma, \tau)$ and have two foci at $\bv{x}^{\pm} = (0, \pm L)$. At any point $\bv{x}$, the coordinates are given by $\tau = \ln \left( d^-/d^+ \right)$, with $d^{\pm} = \lvert \bv{x} - \bv{x}^{\pm} \rvert$, and $\sigma = \mathrm{arg}(\frac{y}{x - L}) - \mathrm{arg}(\frac{y}{x + L})$.
The relationship between the Euclidean $(x, y)$ and bipolar $(\sigma, \tau)$ coordinates is given by
\begin{align} x = \frac{L\sinh \tau}{\cosh \tau - \cos \sigma},\quad y = \frac{L\sin \sigma}{\cosh \tau - \cos\sigma}.
\end{align}
Then, the plastic displacement gradient (\ref{Eq:EdgeDislocationArrayCoframe}) can be simplified in bipolar coordinates as
\begin{align} \bv{\Theta}^2 = \frac{b\rho}{2\pi} \left( -\sigma,\ \tau,\ 0 \right).
\end{align}
Figure \ref{fig:bipolar coordinate}(a) shows the distribution of the plastic deformation fields $\bv{\Theta}^2$ using the bipolar coordinate system.
The same plastic deformation fields are obtained by placing a wedge disclination dipole at the foci of the bipolar coordinate system (see figure \ref{fig:bipolar coordinate}(b)).

\subsection{Wedge disclination monopole}

Next, we consider a semi-infinite edge dislocation array.
As shown in figure \ref{fig:bipolar coordinate}(c), the dislocations are distributed uniformly on the positive side of the $x$-axis.
Then, the torsion 2-form is given by:
\begin{align}\label{eq:Semi-infiniteEdgeDislocationTorsion}
    \bv{\tau}^2=b\rho\delta(y)H(x)dx\wedge dy,
\end{align}
where $\rho$ denotes the edge dislocation number density.
Then, we have the following theorem.

\begin{theorem}[Wedge disclination monopole]
\label{thm:Monopole}
Suppose we have a semi-infinite array of edge dislocations whose torsion is given by equation (\ref{eq:Semi-infiniteEdgeDislocationTorsion}).
Then, a single wedge disclination exists at the coordinate origin whose Frank vector is $\phi=-b\rho$.
\end{theorem}

\begin{proof}
Similar to the previous case, we prove the theorem using an analytical expression of plastic deformation fields.
These are obtained by considering the limits $L^-\to 0$ and $L^+\to \infty$ introduced in the previous solution (\ref{Eq:EdgeDislocationArrayCoframe}).
However, directly taking the limit yields a non-physical divergence in the plastic deformation field, $\Theta^2_2\to\infty$.
To address this mathematical issue, we include an additional term $\frac{b\rho}{2\pi}(0,\ln{L^+},0)$ in the plastic deformation field (\ref{Eq:EdgeDislocationArrayCoframe}) to counteract the inappropriate divergence.
This modification is justified because it does not alter Cartan's first structure equation.
By applying this adjustment and subsequently taking the limits $L^-\to 0$ and $L^+\to \infty$, the resulting non-vanishing plastic displacement gradient for the semi-infinite dislocation array becomes:
\begin{align}
\label{eq:ThetaWedgeMonopole}
    \bv{\Theta}^2= \frac{b\rho}{2\pi}\left(\mathrm{arg} \biggl(\frac{y}{x}\biggr), \ln\sqrt{x^2+y^2}, 0\right),
\end{align}
where $\bv{\Theta}^1=\bv{\Theta}^3=\bv{0}$.
Therefore, the Cartan moving frame becomes $\bv{\vartheta}^2=(dx, dy+\bv{\Theta}^2, dz)$.
This moving frame satisfies Cartan's structure equations (\ref{eq:CartanEquationsTau}) with the Weitzenb\"ock connection.
Subsequently, we insert the plastic deformation field $\bv{\vartheta}^2$ into Cartan's structure equations (\ref{Eq:CartanEquationsCurv}) with the Levi-Civita connection.
From the first structure equation (\ref{Eq:CartanEquationsCurv})$_1$, we obtain
\begin{align}
        {\bv{\omega}^1_2}\wedge (dy+ \bv{\Theta}^2 )-{\bv{\omega}^3_1}\wedge dz=0,\quad
        {\bv{\omega}^1_2}\wedge dx-{\bv{\omega}^2_3}\wedge dz=d\bv{\vartheta}^2,\quad
        {\bv{\omega}^3_1}\wedge dx-{\bv{\omega}^2_3}\wedge (dy+\bv{\Theta}^2 )=0.
\end{align}
Consequently, we have a non-vanishing connection form
\begin{align}\label{eq:Curvature2FormSemiinfiniteArray}
    \bv{\omega}^1_2=-b\rho \delta(y)H(x)dy,\quad
    \bv{\omega}^2_3=\bv{\omega}^3_1=0.
\end{align}
By substituting the corresponding result into the second structure equation (\ref{Eq:CartanEquationsCurv})$_2$, we obtain the curvature responsible for the plastic deformation field $\bv{\vartheta}^2$:
\begin{align}
    \bv{\Omega}^1_2=-b\rho\delta(x,y)dx\wedge dy.
\end{align}
Following the geometric definition of wedge disclinations (\ref{eq:WedgeCurvatureForms}), we can conclude that a single wedge disclination exists at the coordinate origin with Frank vector $\phi= -b\rho$.
Note that the sign of the Frank vector changes if we consider the semi-infinite array on the other side, which is described by the limits $L^-\to \infty$ and $L^+\to 0$.
\end{proof}

As in the previous case, we can represent the plastic deformation fields of the wedge disclination monopole using a log-polar coordinate system, denoted by $(\varrho,\theta)=(\ln\sqrt{x^2+y^2},\mathrm{arg}(y/x))$.
By substituting this relationship into equation (\ref{eq:ThetaWedgeMonopole}), we obtain
\begin{align}
    \bv{\Theta}^2=\frac{b\rho}{2\pi}(\theta, \varrho,0).
\end{align}
Figure \ref{fig:bipolar coordinate}(d) illustrates the plastic deformation field distribution $\bv{\Theta}^2$ resulting from a single wedge disclination expressed in the log-polar coordinate system.
Again, this result is identical to the semi-infinite edge dislocation array shown in figure \ref{fig:bipolar coordinate}(c).
In actual crystalline materials, a semi-infinite edge dislocation array can represent a state in which one end of the dislocation array is included within the material, while the other end is exposed.
Consequently, wedge disclination monopoles could exist within the material.
Indeed, single wedge disclinations have been observed in nanoscale crystals and can be considered analogous to the semi-infinite edge dislocation array configuration~\cite{galligan_fivefold_1972}.
It is well known that many physical phenomena can be represented using the bipolar coordinate system, including magnetic and electric fields.
Notably, magnetic fields cannot be isolated. There are no magnetic monopoles, whereas electric fields can originate from monopoles.
Accordingly, disclinations can be regarded as defects with properties fundamentally similar to those of electric fields.

\subsection{Dipole momentum of wedge disclinations}

Finally, let us consider the case in which an edge dislocation exists independently.
Following the previous examples, let the dislocation line lie along the $z$-axis with Burgers vector $\bv{b}=(0, b, 0)$, as shown in figure \ref{fig:bipolar coordinate}(e).
In this case, the non-vanishing torsion 2-form is given as
\begin{align}\label{eq:SingleEdgeDislocationTorsion}
    \bv{\tau}^2=b\delta(x,y)dx\wedge dy.
\end{align}
Then, we have the following theorem.

\begin{theorem}[Single edge dislocation]
\label{thm:DipoleMomentum}
Suppose we have a single edge dislocation whose torsion is given by equation (\ref{eq:SingleEdgeDislocationTorsion}).
Then, the plastic deformation field around the dislocation is equivalent to that generated by a wedge disclination dipole at an infinitesimal distance.
\end{theorem}

\begin{proof}
By inserting the Hodge dual of torsion 2-form of the single edge dislocation (\ref{eq:SingleEdgeDislocationTorsion}) into the Biot--Savart law (\ref{Eq:Biot-Savart}), the plastic displacement gradient becomes
\begin{align}
\label{eq:SingleEdgeDislocationDisplacementGradient}
    \bv{\Theta}^2= \frac{b}{2\pi} \left( -\frac{y}{x^2+y^2}, \frac{x}{x^2+y^2}, 0\right),
\end{align}
where $\bv{\Theta}^1=\bv{\Theta}^3=\bv{0}$.
Therefore, the Cartan moving frame becomes $\bv{\vartheta}=(dx, dy+\bv{\Theta}^2, dz)$.
Note that this result is essentially equivalent to equation (\ref{Eq:EdgeDislocationCoframe}) because it represents another edge dislocation with Burgers vector $\bv{b}=(b,0,0)$.
By inserting the plastic deformation field $\bv{\vartheta}^2$ into Cartan's first structure equation (\ref{Eq:CartanEquationsCurv})$_1$ with the Levi-Civita connection, we obtain the following:
\begin{align}
        {\bv{\omega}^1_2}\wedge (dy+ \bv{\Theta}^2 )-{\bv{\omega}^3_1}\wedge dz=0,\quad
        {\bv{\omega}^1_2}\wedge dx-{\bv{\omega}^2_3}\wedge dz&=d\bv{\vartheta}^2,\quad
        {\bv{\omega}^3_1}\wedge dx-{\bv{\omega}^2_3}\wedge (dy+ \bv{\Theta}^2 )=0.
\end{align}
Consequently, we have a non-vanishing connection form given by
\begin{align}
    \bv{\omega}^1_2=
    \frac{b^2y\delta(x,y)}{bx+2\pi(x^2+y^2)}dx
    -b\delta(x,y)dy,
    \quad \bv{\omega}^2_3=\bv{\omega}^3_1=0.
\end{align}
Inserting the corresponding result into Cartan's second structure equation (\ref{Eq:CartanEquationsCurv})$_2$, we obtain the curvature responsible for the plastic deformation field $\bv{\vartheta}$:
\begin{align}
\bv{\Omega}^1_2=&
    -\left(b\delta(y)\frac{d}{dx}\delta(x)
    +
    \pd{}{y}\frac{b^2y\delta(x,y)}{bx+2\pi(x^2+y^2)}
    \right)
    dx\wedge dy.
\end{align}
This curvature signifies the presence of wedge disclinations that localized and extend uniformly along the $z$-axis.
The dipole moment $\bv{P}$ of wedge disclinations on the $xy$-plane ~\cite{pretko_fracton-elasticity_2018} becomes
\begin{align}
    \label{eq:dipole moment}
    \bv{P} = \int_{\mathbb{R}^2} \bv{\Omega}^1_2 \,\bv{x} = (b, 0).
\end{align}
According to the definition of the dipole moment, there exists a wedge disclination dipole at an infinitesimal distance along $x$-axis.
This proves the Theorem.
\end{proof}

Notably, the dipole moment $\bv{P}$ in equation (\ref{eq:dipole moment}) is perpendicular to the Burgers vector $\bv{b}=(0,b,0)$.
This observation aligns with the previous study~\cite{pretko_fracton-elasticity_2018}, demonstrating the validity of the above theorem.
Figures \ref{fig:bipolar coordinate}(e) and (f) show the plastic deformation field distribution $\bv{\Theta}^2$ generated by a single edge dislocation and a wedge disclination dipole with an infinitesimal separation distance.
From the geometric equivalence between these two configurations, we can draw an important conclusion: the edge dislocation represents the dipole moment of wedge disclinations.
This analogy closely parallels the concepts of electric polarization in an electric dipole and magnetic spin in a magnetic dipole.
This relationship can also be readily confirmed through a straightforward mathematical limiting process.
As explained previously, equation (\ref{Eq:EdgeDislocationArrayCoframe}) describes the plastic deformation fields of a wedge disclination dipole.
Considering the mathematical limits of the plastic deformation gradient $\bv{\Theta}^2$, we obtain
\begin{align}
    \label{eq:LimitDipole}
    \lim_{L^\pm\to 0} \frac{\Theta^2_1}{N} = -\frac{b}{2\pi}\frac{y}{x^2+y^2},\quad
    \lim_{L^\pm\to 0} \frac{\Theta^2_2}{N} = \frac{b}{2\pi}\frac{x}{x^2+y^2}.
\end{align}
This is precisely the $\bv{\Theta}^2$ of a single edge dislocation, as given in equation (\ref{eq:SingleEdgeDislocationDisplacementGradient}).
This provides an alternative proof to Theorem \ref{thm:DipoleMomentum}.

\section{Mechanical fields of wedge disclinations}

\subsection{Riemannian holonomy for Frank vector analysis}

Riemann--Cartan manifolds represent plastic deformation fields $\bv{\vartheta}$ in two distinct forms: the Weitzenb\"ock $(\mathcal{M}, g, \nabla^W)$ and Riemannian manifolds $(\mathcal{M}, g, \nabla^L)$.
These manifolds are related by replacing the connection while preserving the Riemannian metric $g$.
Because of the inherent arbitrariness of the connection, we demonstrated that edge dislocations and wedge disclinations are geometrically related rather than independent entities.
This implies that standard geometric analysis on a Riemannian manifold can be applied to the study of dislocations defined on a Weitzenb\"ock manifold.
A particularly attractive example is Riemannian holonomy, which is a generalization of parallel transport for vectors in a curved space.

Let us consider the parallel transportation of a vector $\bv{X}$ along a smooth and closed curve $\bv{c}=\bv{c}(t)$ on a Riemannian manifold.
Parallel transportation is defined by the following differential equation \cite{nakahara_geometry_2003}:
\begin{align}
  \label{eq:vector parallel transportation}
     \frac{d\bv{X}}{dt}+\omega^i_{jk}\frac{dc^j}{dt} X^k \bv{e}_i=0,
\end{align}
where $\omega^i_{jk}$ is the coefficient of the Levi-Civita connection 1-form defined by $\bv{\omega}^i_j=\omega^i_{jk}d\bv{x}^k$.
We express the initial and final vectors of the parallel transportation by $\bv{X}_s$ and $\bv{X}_e$, respectively.
These vectors are related by the linear transformation $\bv{X}_e= \bv{\Phi} \bv{X}_s$, where $\bv{\Phi}$ is the Frank matrix, whose components are expressed in the following form~\cite{puntigam_volterra_1997,fumeron_introduction_2023}:
\begin{align}
  \label{eq:FrankMatrix}
  \Phi^i_j=\mathcal{P} \exp\biggl(\oint_c -\bv{\omega}^i_{j}\biggr).
\end{align}
In this equation, $\mathcal{P}$ is the path-ordering operator.

We apply the Riemannian holonomy analysis to the edge dislocation configurations, as shown in figure \ref{fig:bipolar coordinate}(a).
Mathematically, this configuration is expressed by the Weitzenb\"ock manifold as it expresses the dislocations by torsion through the connection $\nabla^W$.
According to Theorem \ref{thm:Dipole}, however, this configuration is geometrically equivalent to a disclination dipole (see figure \ref{fig:bipolar coordinate}(b)), which is expressed as a Riemannian manifold.
Let $c^\pm$ be a closed curve encircling the foci located at $\bv{x}^\pm=(\pm L^\pm,0)$.
From the analytical Riemannian connection given in equation (\ref{eq:Curvature2FormFiniteArray}), we obtain the Frank matrix for the two curves $c^\pm$, such that:
\begin{align}
  \bv{\Phi}(c^+)=\begin{pmatrix}
    \cos \phi &-\sin \phi &0\\
    \sin \phi &\cos \phi &0\\
    0&0&1
  \end{pmatrix},
  \quad
  \bv{\Phi}(c^-)=\begin{pmatrix}
    \cos \phi &\sin \phi &0\\
    -\sin \phi &\cos \phi &0\\
    0&0&1
  \end{pmatrix}.
\end{align}
This result indicates that when a vector $\bv{X}_s$ in the $xy$-plane undergoes parallel transport, it experiences an angular change of $\phi$ after the trandportation to $\bv{X}_e$, irrespective of the specific closed curve $c^+$ encircling the focal point located at $\bv{x}^+$.
Performing the same operation along a closed curve $c^-$ surrounding $\bv{x}^-$ results in an angular change in $-\phi$.
This result aligns with the properties of the wedge disclination shown in figure~\ref{fig:VolterraDefects}(f).
Furthermore, the magnitude of the rotation angle $\phi$ is quantitatively consistent with the results of Theorem \ref{thm:Dipole}.
Therefore, we can conclude that a wedge disclination dipole exists at both ends of the edge dislocation array (figures \ref{fig:bipolar coordinate}(a) and (b)).
Additionally, by applying a similar analysis to a semi-infinite edge dislocation array, we can conclude that a wedge disclination monopole exists at the coordinate origin.
These are also mathematical proofs of Theorems \ref{thm:Dipole} and \ref{thm:Monopole} using the Riemannian holonomy.
Finally, in the holonomy analysis surrounding a single edge dislocation, the Frank matrix becomes the identity matrix because curve $c$ always encloses a wedge disclination dipole, causing the rotation angles to cancel out.

\subsection{Cauchy--Riemann equations and complex potential}

The geometric equivalence between edge dislocations and wedge disclinations introduces another remarkably new framework: complex function analysis for disclinations.
In a previous study~\cite{kobayashi_biot--savart_2024}, we demonstrated that Cartan's first structure equation and Helmholtz decomposition for the plastic deformation field $\bv{\vartheta}$ associated with dislocations are mathematically equivalent to two distinct sets of equations from other disciplines: electromagnetic field equations and the Cauchy--Riemann equations from complex function analysis.
The latter describes the conformal properties of plastic deformation fields generated around dislocations.
Using the mathematical framework, we demonstrated that the plastic deformation fields of wedge disclinations can be derived from a single complex potential~\cite{kobayashi_biot--savart_2024}.
Considering the linearity of Cartan's first structure equation with the Weitzenb\"ock connection and the geometrical equivalence established via Theorems \ref{thm:Dipole} and \ref{thm:Monopole}, we can define the complex potential for wedge disclinations.

First, we examine the wedge disclination dipole, as shown in figure \ref{fig:bipolar coordinate}(b).
The orthogonality of the plastic displacement fields $\bv{\Theta}^2$ given in equation (\ref{Eq:EdgeDislocationArrayCoframe}) is obvious because it can be expressed in bipolar coordinates $(\sigma, \tau)$, which is an orthonormal coordinate system.
This remarkable mathematical property allows the introduction of the complex potential for the wedge disclination dipole.

\begin{theorem}[Plastic potential for the disclination dipole]
\label{thm:PlasticPotentialDipole}
Let $\Psi^D$ be a complex function of the form
\begin{align}
\label{eq:PlasticPotentialDipole}
    \Psi^D(z)=-\frac{\mathrm{i} b\rho}{2\pi}
    \left(
        (L+z)\ln(-(L+z)) + (L-z)\ln(L-z)
    \right),
\end{align}
defined on the complex plane as $z = x + \mathrm{i}y$.
Then the plastic displacement gradients of the wedge disclination dipole are obtained from the potential function such that $\Theta^2_1 = \mathrm{Re}( d \Psi^D/d z )$ and $\Theta^2_2 = -\mathrm{Im}( d \Psi^D/d z )$.
\end{theorem}

\begin{proof}
The partial derivatives of the complex potential are $\partial \Psi^D/\partial x = -(\mathrm{i}b\rho/2\pi)(\ln(-(L+z))-\ln(L-z))$ and $\partial \Psi^D/\partial y =(b\rho/2\pi)(\ln(-(L+z))-\ln(L-z))$.
By definition, the differential operator with respect to complex variable $z$ is $d/dz=(\partial/\partial x-\mathrm{i}\partial/\partial y)/2$.
Therefore, we have
\begin{align}
    \frac{d\Psi^D}{dz}
    =\frac{1}{2}\left(
    \frac{\partial }{\partial x}-\mathrm{i}\frac{\partial}{\partial y}
    \right)\Psi^D
    = -\frac{\mathrm{i}b\rho}{2\pi}(\ln(-(L+z))-\ln(L-z)).
\end{align}
Through direct calculations, we can confirm that $\mathrm{Re}( d \Psi^D/d z )$ and $-\mathrm{Im}( d \Psi^D/d z )$ correspond to $\Theta^2_1$ and $\Theta^2_2$ of the plastic displacement fields given in equation (\ref{Eq:EdgeDislocationArrayCoframe}).
\end{proof}

\begin{theorem}[Plastic potential for the disclination monopole]
\label{thm:PlasticPotentialMonopole}
Let $\Psi^M$ be a complex function of the form
\begin{align}
\label{eq:PlasticPotentialMonopole}
    \Psi^M(z)=-\frac{\mathrm{i} b\rho}{2\pi}z(\ln(-z)-1),
\end{align}
defined on the complex plane as $z = x + \mathrm{i}y$.
Then the plastic displacement gradients of the wedge disclination monopole is obtained from the potential function such that $\Theta^2_1 = \mathrm{Re}( d \Psi^M/d z )$ and $\Theta^2_2 = -\mathrm{Im}( d \Psi^M/d z )$.
\end{theorem}

\begin{proof}
The partial derivatives of the potential are $\partial \Psi^M/\partial x = -(\mathrm{i} b\rho/2\pi)\ln{(-z)}$ and $\partial \Psi^M/\partial y = ( b\rho/2\pi)\ln{(-z)}$.
Hence, we have
\begin{align}
\frac{d\Psi^M}{dz}
=\frac{1}{2}\left(
\frac{\partial }{\partial x}-\mathrm{i}\frac{\partial}{\partial y}
\right)\Psi^M
=-\frac{\mathrm{i} b\rho}{2\pi}\ln{(-z)}.
\end{align}
Through direct calculations, we can confirm that $\mathrm{Re}( d \Psi^M/d z)$ and $-\mathrm{Im}( d \Psi^M/d z)$ correspond to $\Theta^2_1$ and $\Theta^2_2$ of the plastic displacement fields given in equation (\ref{eq:ThetaWedgeMonopole}).
\end{proof}

Figures \ref{fig:complex potentials}(a) and (b) show the real and imaginary parts of the complex potential of the disclination dipole, respectively.
The imaginary part $\mathrm{Im}(\Psi^D)$ is a single-valued continuous function that decreases monotonically with increasing distance from a dislocation array.
Meanwhile, the real part $\mathrm{Re}(\Psi^D)$ contains two branch points $z=\pm L$ corresponding to the positions of the disclination dipole.
A notable feature here is the jump discontinuity along the line $\mathcal{L}^D=\{x+\mathrm{i}y\mid y=0,\ x> -L\}$.
Mathematically, this represents a branch cut, indicating that the potential is multivalued.
The jump height on the branch cut $\llbracket \mathrm{Re}(\Psi^D) \rrbracket_{\mathcal{L}^D}$ is calculated as follows:
\begin{align}
    \label{eq:jump height dipole}
    \llbracket \mathrm{Re}(\Psi^D) \rrbracket_{\mathcal{L}^D} =\begin{cases}
        -\dfrac{bN}{2L}(x+L) & -L< x < L\\
        -bN & x\geq L
    \end{cases},
\end{align}
where $N$ is the number of dislocations in the array.
The jump height $\llbracket \mathrm{Re}(\Psi^D) \rrbracket_{\mathcal{L}^D}$ can be interpreted as the magnitude of the extra half-planes associated with the edge dislocations, as illustrated in figure \ref{fig:bipolar coordinate}(a).
Specifically, the constant $-bN$ on $x\geq L$ represents the total magnitude of the Burgers vector, while a linear variation occurs in the range $-L<x<L$.
The jump discontinuity encapsulates the topological properties of the disclination dipole.
Additionally, the limit $L\to 0$ gives $\lim_{L\to 0} \Psi^D/N=-(\mathrm{i}b/2\pi)(1+\ln(-z))$, which aligns with the complex potential of a single dislocation given in our previous study up to the sign and an additive constant~\cite{kobayashi_biot--savart_2024}, which in turns aligns with the observation in equation (\ref{eq:LimitDipole}).

A similar discussion holds for the case of a disclination monopole, as shown in figures \ref{fig:complex potentials}(c) and (d).
The imaginary part $\mathrm{Im}(\Psi^M)$ of the disclination monopole is a continuous function that is asymmetrically distributed with respect to the imaginary axis, which corresponds to the asymmetric arrangement of the dislocation array.
Meanwhile, the real part $\mathrm{Re}(\Psi^M)$ is a multivalued function with a branch point at $z=0$ and a branch cut along the line $\mathcal{L}^M=\{x+\mathrm{i}y \mid y=0,\ x>0\}$ containing the jump discontinuity.
The jump height along the branch cut is $\llbracket \mathrm{Re}(\Psi^M) \rrbracket_{\mathcal{L}^M}=-b\rho x$, which forms a linear function similar to the disclination dipole described in equation (\ref{eq:jump height dipole}).
This demonstrates that the complex potential encapsulates the topological change due to defects, as revealed for the disclination dipole case.
Furthermore, direct calculations confirm that the complex potential of the disclination dipole (\ref{eq:PlasticPotentialDipole}) can be obtained by the superposition of those of the disclination monopole up to a constant: $\Psi^D(z)=\Psi^M(z+L)-\Psi^M(z-L)+\mathrm{const}.$

\begin{figure}[!h]
    \centering
    \includegraphics{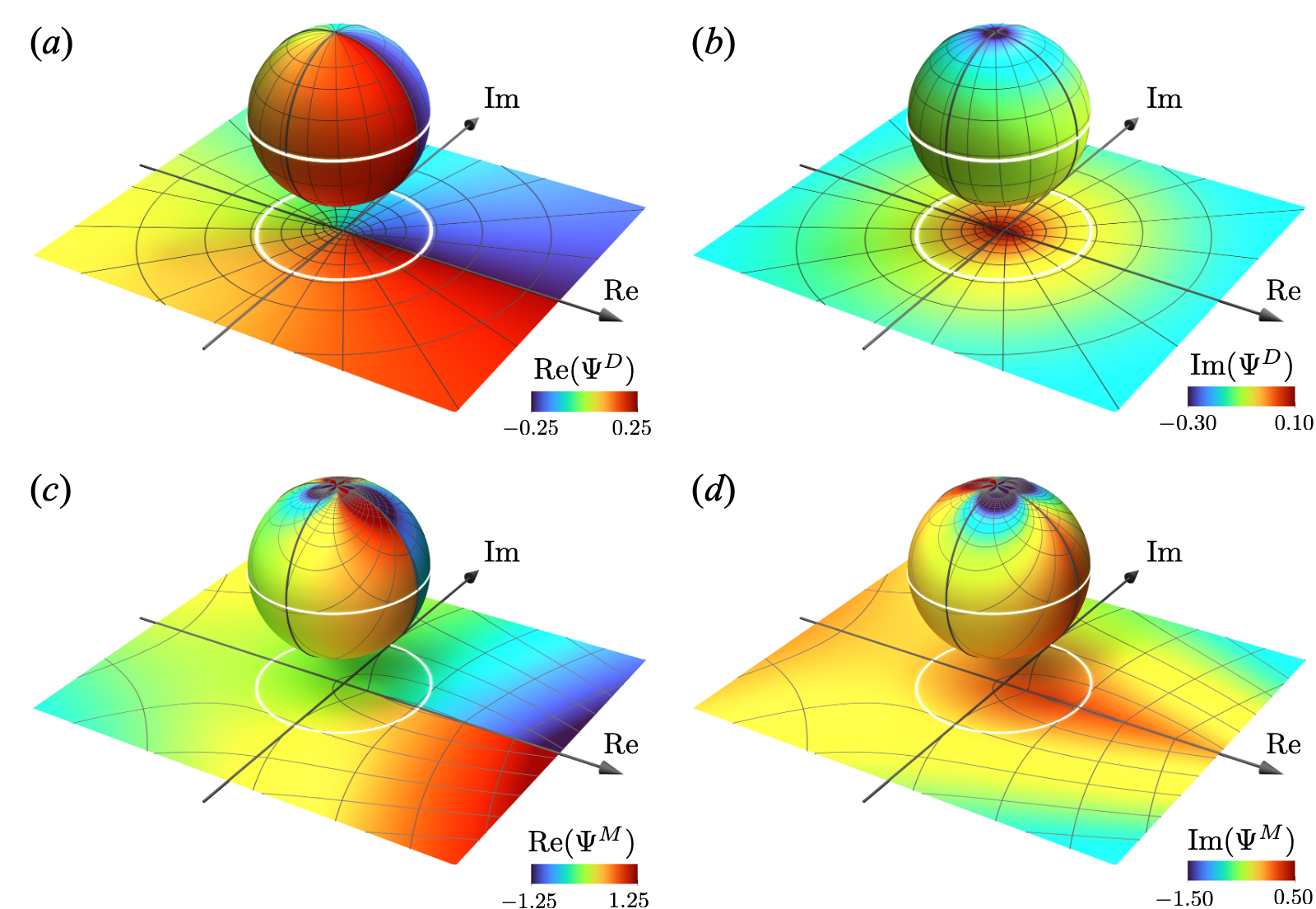}
    \caption{
        Plastic deformation potential of the disclination dipole $\Psi^D$ and monopole $\Psi^M$, plotted on Riemann spheres and their stereographic projections onto the complex planes.
        The projection $(x^1, x^2, x^3)\in\mathbb{R}^3 \mapsto (x+\mathrm{i}y)\in\mathbb{C}$ is defined by a map $x+\mathrm{i}y=(x^1+\mathrm{i}x^2)/(1+x^3)$.
        Here, the equator of the sphere is projected as a white circle onto the complex plane.
        (a) and (b) Real and imaginary parts of the complex potential $\Psi^D$ of the disclination dipole with $b=\rho=1$ and $L=1/4$.
        (c) and (d) Real and imaginary parts of the monopole $\Psi^M$ with $b=\rho=1$.
    }
    \label{fig:complex potentials}
\end{figure}

\subsection{Stress fields of the wedge disclinations}
\label{sec:stress disclination}

Finally, we examine the elastic stress fields associated with disclinations.
Current geometric theory constructs a plastically deformed state on a Riemann--Cartan manifold. Mathematically, these states are incompatible with Euclidean geometry; therefore, the manifold cannot be embedded in Euclidean space in its original form.
To address this issue, it is necessary to introduce a compensating elastic deformation that resolves the incompatible geometric frustration~\cite{kobayashi_geometrical_2024}.
Generally, the stress equilibrium equation for elasticity is non-linear, rendering an analytical solution unfeasible.
However, by linearizing this equation, we can analytically construct an elastic stress field for specific defect configurations~\cite{kobayashi_biot--savart_2024}.
In this study, we apply the same approach to analyse the stress fields of wedge disclinations.

Linearization of the equilibrium equations requires linearization of the kinematics using the Cauchy strain~\cite{kobayashi_biot--savart_2024}.
First, we introduce the total, plastic, and elastic Cauchy strains, defined by
\begin{align}\label{eq:lin cauchy strains}
    \bv{\mathcal{E}}_t={}\frac{1}{2}(\bv{\nabla u}+\bv{\nabla u}^T),
    \quad
    \bv{\mathcal{E}}_p={}\frac12(\bv{\Theta}+\bv{\Theta}^T),
    \quad
    \bv{\mathcal{E}}_e=\bv{\mathcal{E}}_t-\bv{\mathcal{E}}_p.
\end{align}
The linearized stress can be expressed using Hooke's law as $\bv{\sigma}=\bv{C}:\bv{\mathcal{E}}_e$, where $\bv{C}$ is a stiffness tensor.
The coefficients of $\bv{C}$ are expressed using the Poisson ratio $\nu$ and shear modulus $\mu$ as $C^{ijkl}=\frac{2\nu\mu}{1-2\nu}\delta^{ij}\delta^{kl} + \mu (\delta^{ik}\delta^{jl}+\delta^{il}\delta^{jk})$.
Subsequently, the linearized stress equilibrium equation is given by $\nabla\cdot\bv{\sigma}=\bv{0}$. In terms of the plastic $\bv{\mathcal{E}}_p$ and total strains $\bv{\mathcal{E}}_t$, it can be rewritten as follows:
\begin{align}
    \label{eq:stress equil}
    \nabla\cdot(\bv{C}:\bv{\mathcal{E}}_t)=\nabla\cdot (\bv{C}:\bv{\mathcal{E}}_p).
\end{align}
By inserting the plastic displacement gradient $\bv{\Theta}^2$ of the wedge disclination dipole (\ref{Eq:EdgeDislocationArrayCoframe}) into the plastic strain definition (\ref{eq:lin cauchy strains})$_2$ we have the following:
\begin{align}
    \bv{\mathcal{E}}_p = 
    \frac{b \rho}{4\pi}
    \begin{pmatrix}
        0&-\sigma&0\\
        -\sigma& 2\tau&0\\
        0&0&0
    \end{pmatrix},
\end{align}
where $(\sigma, \tau)$ are the bipolar coordinates.
By inserting the result into the right-hand side of equation (\ref{eq:stress equil}) and using the convolution integration of the edge dislocation displacements reported in our previous study~\cite{kobayashi_biot--savart_2024},
we obtain the linearized total displacement field $\bv{u}$ as
\begin{align}
    \bv{u}
    =\frac{D}{2\mu}
    \left(
        (3-4\nu)L
        -2y(1-\nu) \sigma-(1-2\nu)((x+L)\ln d^{-}-(x-L)\ln d^{+}),
        y\tau,
        0
    \right),
\end{align}
where $D=\frac{\mu b \rho}{2\pi(1-\nu)}$.
Consequently, we obtain the linearized stress fields of the wedge disclination dipole as follows:
\begin{align}
    \bv{\sigma}=&{}
    -D
    \begin{pmatrix}
        \dfrac{y^2}{(d^{-})^2} - \dfrac{y^2}{(d^+)^2} + \tau &
        -\dfrac{(x+L)y}{(d^{-})^2}+\dfrac{(x-L)y}{(d^+)^2} &
        0\\
        -\dfrac{(x+L)y}{(d^{-})^2}+\dfrac{(x-L)y}{(d^{+})^2} &
        -\dfrac{y^2}{(d^{-})^2} + \dfrac{y^2}{(d^{+})^2} + \tau &
        0 \\
        0 & 0 & 2\nu \tau \\
    \end{pmatrix}.
\end{align}
The results are in complete agreement with those of the previous study~\cite{li_disclination_1972}.

The stress fields of the wedge dislocation monopole are determined using the same method.
In terms of log-polar coordinates $(\varrho,\theta)$, the plastic Cauchy strain $\bv{\mathcal{E}}_p$ of the monopole becomes
\begin{align}
    \bv{\mathcal{E}}_p = 
    \frac{b \rho}{4\pi}
    \begin{pmatrix}
        0&\theta&0\\
        \theta& 2\varrho&0\\
        0&0&0
    \end{pmatrix}.
\end{align}
After convolution integration with introducing additional terms to cancel out uniformly diverging terms as in the proof of Theorem \ref{thm:Monopole}, we obtain the total displacement, such that
\begin{align}
    \bv{u}
    =&{}
    \frac{D}{2\mu}
    \left(
        2(1-\nu) y\theta-(1-2\nu)x(\varrho-1),
        y(\varrho-1),
        0
    \right).
\end{align}
By substituting this into the stress equilibrium equation and using Hooke's law, the elastic Cauchy strain and linearized stress are obtained as follows:
\begin{align}
    \bv{\sigma}=&{}
    -D
    \begin{pmatrix}
        \dfrac{y^2}{r^2}+\varrho &
        -\dfrac{x y}{r^2} &
        0 \\
        -\dfrac{x y}{r^2} &
        -\dfrac{y^2}{r^2}+\varrho &
        0 \\
        0 & 0 & 2\nu\varrho \\
    \end{pmatrix}.
\end{align}
Again, these results agree up to constants with those of a previous study~\cite{dewit_theory_1973}.

The coincidence of the elastic stress fields has two major implications.
The first is the validation of the mathematical analysis presented in this study. Because theoretical analyses of Volterra defects are limited, objective validation using previous reports is challenging. To address this issue, this study adopted a mathematical approach, formulating the analysis as a theorem with rigorous proof.
Nevertheless, consistency of the results with existing analyses is essential to demonstrate the applicability of our findings.
The second implication is the potential extension of this theory to non-linear mechanics. The agreement between the linearized theory and existing results establishes the present framework as a natural extension of the conventional theory into the non-linear mechanics domain. This theoretical framework is expected to drive further advancements in disclination analysis.

\section{Conclusion}

In this study, we developed a mathematical model for Volterra defects using differential geometry on Riemann--Cartan manifolds and systematically examined their relationships by analytically solving plastic deformation fields. Based on the results, the main conclusions of this study can be summarized as follows.

\renewcommand{\labelenumi}{(\arabic{enumi})}
\begin{enumerate}
    \item 
    We introduced Volterra deformations as translational and rotational deformations with respect to the three coordinate axes and defined Cartan's moving frame as a mathematical representation of the plastic deformation field. For dislocation analysis, we applied the Weitzenb\"ock connection to Cartan's structure equations, revealing that the dislocation density aligns with the classical definition in the lattice defect theory. Similarly, we examined disclinations using Cartan's structure equations. Although curvature naturally appears in all three cases, excess torsion or curvature components inevitably persist when modelling twist disclinations. This suggests that modifications are required in the Volterra process itself.
    \item 
    From a mathematical perspective, the connection choice in the Riemann--Cartan manifold is not unique, allowing for the replacement of the Weitzenb\"ock and Levi-Civita connections. This mathematical flexibility establishes a geometric equivalence between dislocations and disclinations as topological defects. To demonstrate this, we examined specific cases where the existence of wedge disclinations has previously been suggested phenomenologically. By leveraging connection replacement and analytical solutions for plastic deformations derived via the Biot--Savart law, we provide a rigorous mathematical proof of the existence of wedge disclinations at the terminal points of the edge dislocation array. This finding further clarifies the geometric relationship between the two topological defects: the edge dislocation serves as the momentum of the wedge disclination dipole. Similarly, we proved that isolating a wedge disclination monopole is geometrically feasible using a semi-infinite edge dislocation array. Furthermore, we revealed that the plastic deformation fields can be represented by an orthogonal coordinate system. This result indicates that plastic deformation fields are inherently conformal.
    \item 
    We analysed the plastic deformation fields of wedge disclinations from multiple perspectives. First, we demonstrated that Riemannian holonomy, a generalization of parallel transport for vectors, can quantitatively measure the Frank vector of a disclination. By leveraging the mathematical equivalence between Cartan's structure equations for plasticity and the Cauchy--Riemann equations in complex function theory, we constructed complex potentials for the plastic deformation of wedge disclinations, elucidating their topological properties, including the jump discontinuity. Finally, we performed a stress field analysis for wedge disclinations. After applying geometric and constitutive linearization, we obtained analytical expressions for the stress fields, which were consistent with the findings of previous studies.
\end{enumerate}

\section*{Acknowledgment}
This work was supported by JST, PRESTO, Japan (Grant Number JPMJPR1997) and JSPS KAKENHI (Grant Numbers JP23K13221 and JP24K00761).

\section*{Data Accessibility}
Data and relevant codes are accessible through Dryad \cite{dryad} and Zenodo \cite{zenodo} repositories, respectively.


\end{document}